\documentclass{siamltex}

\let\pdfximagebbox\undefined
\usepackage{epsfig}
\usepackage{graphicx}
\usepackage{graphics}
\usepackage{subfigure}
\usepackage{amsmath,amsfonts}
\usepackage{textcomp}
\usepackage{color}
\usepackage{mathrsfs}
\usepackage{geompsfi-pdf2}
\usepackage{afterpage}

\def\Re{\mathbb R}
\let\R\Re

\def\Lebesgue{\mathscr L}
\def\L{\mathcal L} 
\let\epsilon\varepsilon
\let\e\varepsilon
\def\HS{H}
\def\H{\mathcal H} 

\title{Self-similar voiding solutions of a single layered model of folding rocks} 

\author{T. J. Dodwell\footnotemark[1]\ 
\and M. A. Peletier\footnotemark[2]\ 
\and C. J. Budd\footnotemark[1]\
\and G. W. Hunt\footnotemark[1]}

\date{\today}

\begin{document}
\maketitle

\renewcommand{\thefootnote}{\fnsymbol{footnote}}

\footnotetext[1]{Bath Institute of Complex Systems, University of Bath, BA2 7AY}
\footnotetext[2]{Institute of Complex Molecular Systems and Department of Mathematics and Computer Science, Technische Universiteit Eindhoven, PO Box 513, 5600MB Eindhoven, The Netherlands}
\footnotetext[3]{Corresponding author: C. J. Budd (mascjb@bath.ac.uk)}
\footnotetext{\date}

\renewcommand{\thefootnote}{\arabic{footnote}}

%
%
\newenvironment{remark}%
  {\par\medbreak\refstepcounter{theorem}%
    \noindent\textit{Remark~\thetheorem. }}%
  {\par\medskip}

\begin{abstract}
In this paper we derive an obstacle problem with a free boundary to describe the formation of voids at areas of intense geological folding. An elastic layer is forced by overburden pressure against a V-shaped rigid obstacle. Energy minimization leads to representation as a nonlinear fourth-order ordinary differential equation, for which we prove their exists a unique solution. Drawing parallels with the Kuhn-Tucker theory, virtual work, and ideas of duality, we highlight the physical significance of this differential equation. Finally we show this equation scales to a single parametric group, revealing a scaling law connecting the size of the void with the pressure/stiffness ratio. This paper is seen as the first step towards a full multilayered model with the possibility of voiding.
\end{abstract}



\pagestyle{myheadings}
\thispagestyle{plain}
\markboth{Dodwell. T J. \textit{et al.}}{Self-similar voiding solutions of folding rocks}

\section{Introduction}\label{sec:Introduction}

The bending and buckling of layers of rock under tectonic plate movement has played a significant part in the Earth's history, and remains of major interest to mineral exploration in the field. The resulting folds are strongly influenced by a subtle mix of geometrical restrictions, imposed by the need for layers to fit together, and mechanical constraints of bending stiffness, inter-layer friction and worked done against overburden pressure in voiding. An example of such a fold is seen in Figure~\ref{fig:millockhaven}, here the voiding is visible through the intrusion of softer material (dark in this figure) between the harder layers (shinier in the figure) which have separated while undergoing intense folding.


\begin{figure}[hbt]
\centering
\includegraphics[height=4cm]{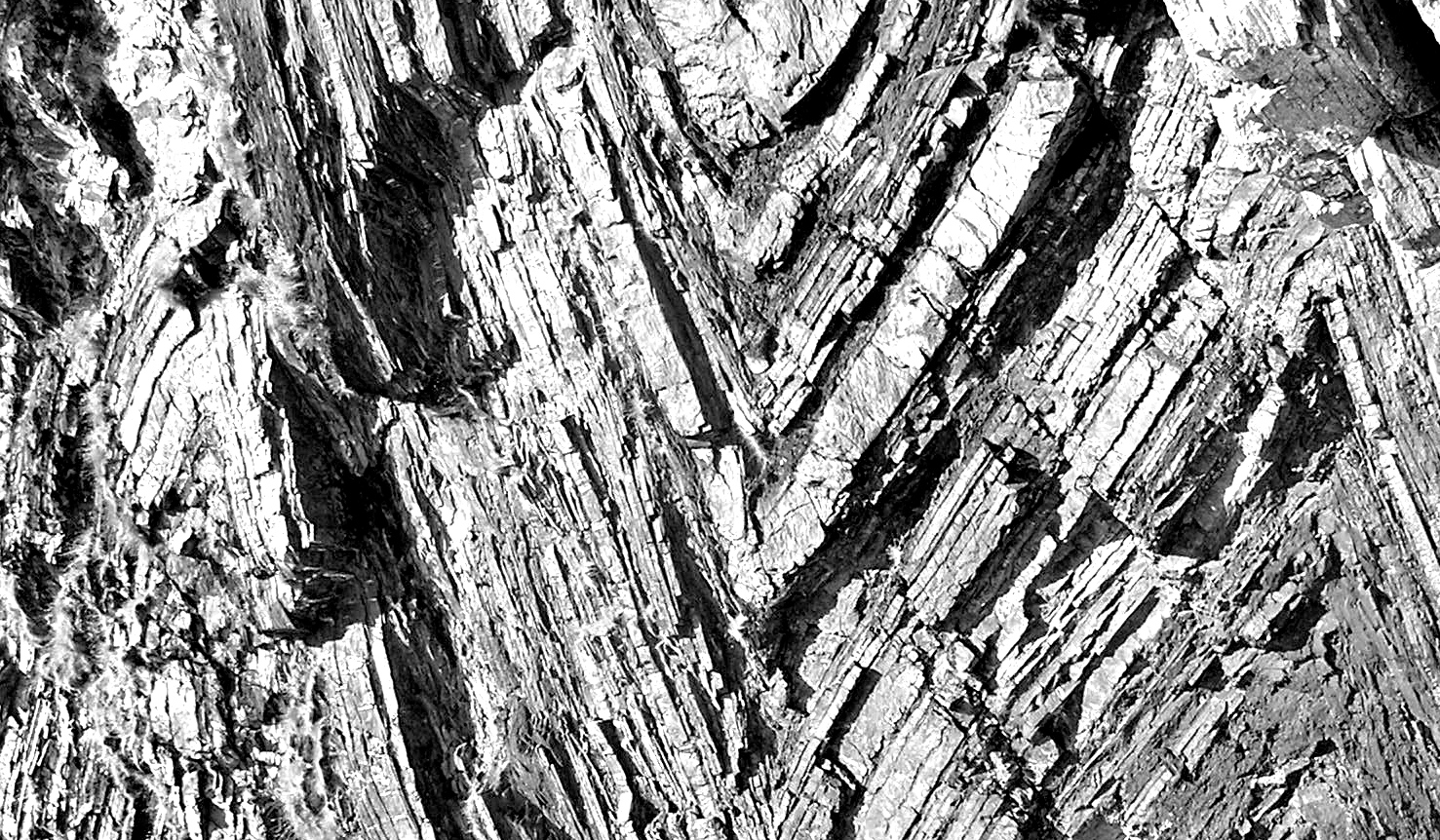}
\caption{A photograph of a geological formation from Millock Haven, Cornwall, UK, demonstrating the formation of voids, visible by the intrusion of softer material, while harder layers undergo intense geological folding. Scale is approximately 5m across.}
\label{fig:millockhaven}
\end{figure}

Consider a system of rock layers, of constant thickness, initially lying parallel to each other that are then buckled by an external horizontal force, while being held together by an overburden pressure. If rock layers do not separate during the buckling process it is then inevitable that sharp corners will develop. To see this, consider a single layer buckled into the shape of a parabola with further layers, of constant thickness, lying on top of this. 
Moving from the bottom layer upwards, geometrical constraints mean that the curvature of the individuals layer tightens until it becomes infinite, marking the presence of a \emph{swallowtail singularity}~\cite{Boon2007}. Beyond this singularity the layers interpenetrate in a non-physical manner. This process is illustrated in Figure~\ref{fig:BBHFig3}, showing how the layers would continue through the singularity if they were free to interpenetrate.  

\begin{figure}[hbt]
\centering
\includegraphics[height=4cm]{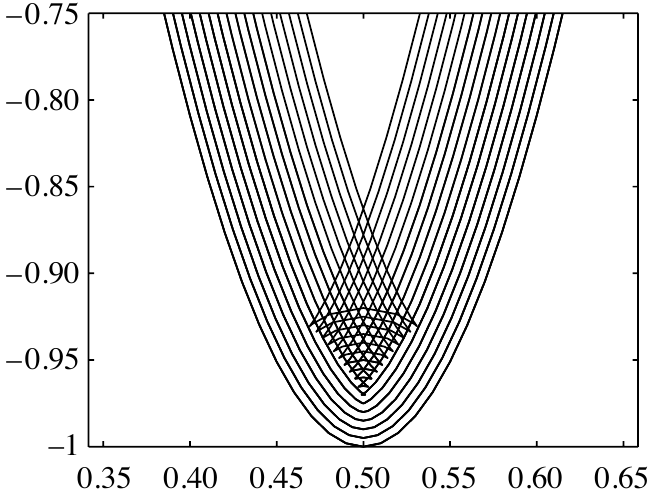}
\caption{A close-up view of the propagation of a sine wave, demonstrating the physically-unrealisable swallowtail catastrophe.}
\label{fig:BBHFig3}
\end{figure}

Models have dealt with these singularities by for instance limiting the number of layers~\cite{Budd2003,Edmunds2006}, using the concept of viscosity solutions~\cite{Boon2007}, or postulating a simplified geometry of straight limbs punctuated by sharp corners, as is observed in kink banding~\cite{HuntPeletierWadee00,Wadee2004}. These approaches, however, disregard the resistance of the layers to bending, which is expected to be especially relevant close to the singularity. Here we therefore introduce the property of  \emph{elastic stiffness} into the modelling, and combine it with a condition of non-interpenetration. As a result, the layers will not fit together completely, but do work against overburdern pressure and create \emph{voids}.The folding of rocks is a complex process with many interacting
factors. In a multilayered model it is clear that work needs to
be done to slide the layers over each other in the presence
of friction, to bend the individual layers and finally to separate
the layers (voiding). In order to understand the interaction between
the process of bending and voiding we will not consider the effects
of friction in this paper but will leave this to the subject of later work.

The process of voiding is illustrated in Figure~\ref{fig:introductionpicture}, which shows a laboratory recreation of folding rocks 
obtained by compressing laterally confined layers of paper. As we move through the sample, the curvature of the layers increase until a point is reached where the work against pressure in voiding balances the work in bending
and the layers separate. A number of features of the voiding process can be seen in this figure. It is clear that the voids have a regular and repeatable form and that a typical void occurs when a smooth layer of paper
separates from one which has a near corner-shape. 

\begin{figure}[bp]
\centering
\includegraphics[height=5cm]{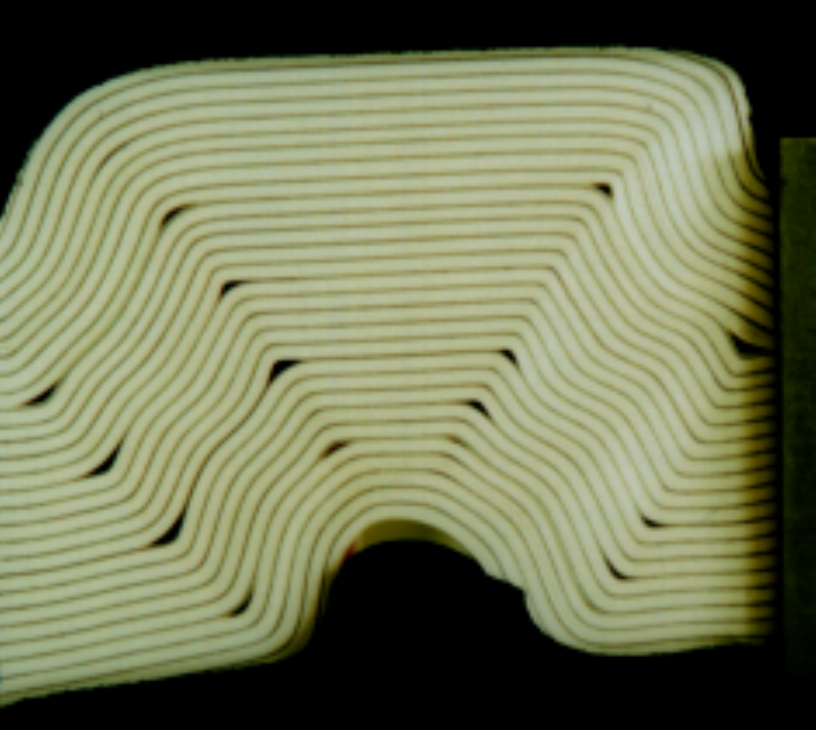}
\caption{ A laboratory experiment of layers of paper constrained and loaded. In this figure the black lines are for illustrative
purposes, and are produced by inserting a single black layer of paper between 25 layers of white. The resulting deformation shows the formation of voids when the imposed curvature becomes too high. Note the regular and repeatable nature of the voids.}
\label{fig:introductionpicture}
\end{figure}

In this paper we present a simplified energy-based model of voiding inspired by the processes observed in Fig.~\ref{fig:introductionpicture}. The model consists of a single elastic layer with a vertical displacement $w(x)$ forced downwards, and bent, into a corner-shaped obstacle 
of shape $f(x) \le  w(x)$ by a uniform overburden pressure~$q$ (see Fig.~\ref{fig:model}). The corner is defined to have infinite curvature at the point, $x = 0$. For $|x|$ sufficiently large, the layer and obstacle
are in contact so that $w = f$. However, close to $x = 0$ the layer and obstacle separate, leading to a single void for those values of $x$ for which $w > f$.  We study both the resulting shape of this elastic layer and the size of the voiding region. This investigation is the first
part of a more general study of the periodic multi-layered voiding pattern seen in Fig.~\ref{fig:introductionpicture}.

To study this situation we construct a potential energy functional $V(w)$ for the system, derived in Section~\ref{sec:Model}, which is given in terms of the vertical displacement $w(x)$ and combines the energy $U_B$ required to {\em bend} the elastic layer and the energy $U_V$ required to {\em separate} adjacent layers and form voids. The potential energy function is then given by
\begin{equation}
{
V = U_B + U_V \equiv \frac{B}{2}\int^{\infty}_{-\infty} \frac{w_{xx}^2}{(1 + w_x^2)^{5/2}} dx + q\int^{\infty}_{-\infty}(w-f)\,dx, \qquad \mbox{where } w\geq f. 
}
\label{eq:IntroductionEnergy}
\end{equation} 
The resulting profile is then obtained by finding the minimiser of $V$ over all suitably regular functions $w \ge f$. This constrained minimization problem is closely related to many other obstacle problems, as can be found in the study of fluids in porous media, optimal control problems, and the study of elasto-plasticity~\cite{Elliot1982}. 

\begin{figure}[htbp]
	\centering
\psfig{figure=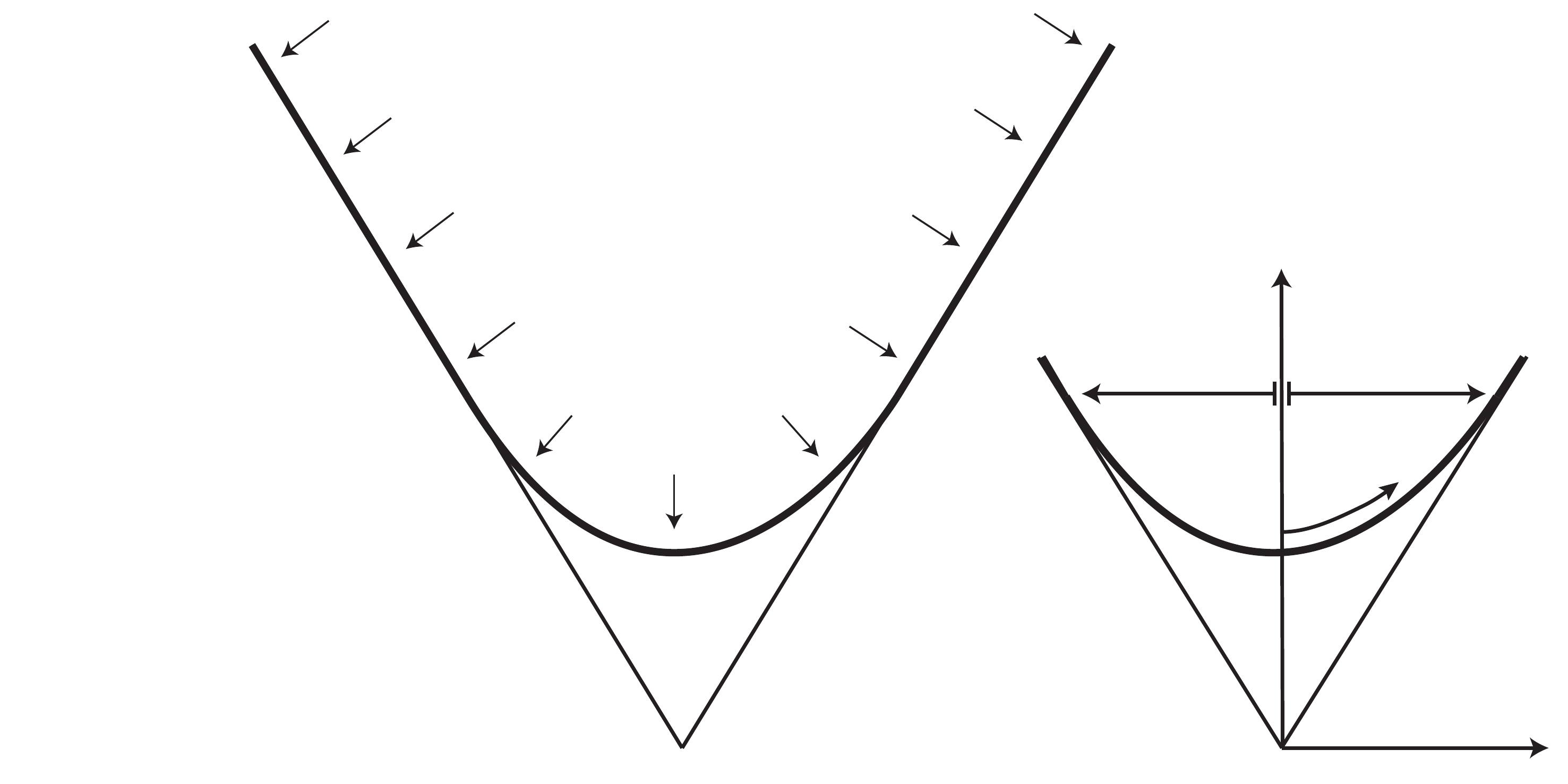,width=4in}		
\caption{This figure shows the setup of the model discussed in this paper. An overburden pressure $q$ forces an elastic layer $w$ into another layer $f$ with a corner singularity. $\ell_+$ and $\ell_-$ define the first points of contact either side of the centre line. In this paper the layer is described both by Cartesian coordinates $(x,w)$ measured from the centre of the singularity, and intrinsic coordinates characterised by arc length $s$ and angle $\psi$.}
	\label{fig:model}
\end{figure}

While obstacle problems are often cast as variational inequalities~\cite{Kinderlehrer1980}, here we use the Kuhn-Tucker theorem for its suitability when interpreting the results physically. In Section~\ref{sec:Model} we prove various qualitative properties of constrained minimizers, and use the Euler-Lagrange equation to derive a fourth-order free-boundary problem that they satisfy. 

In addition, we show that stationarity implies that a certain quantity (the `Hamiltonian') is constant in any region of non-contact (Section~\ref{sec:physical}). This property extends the well-known property of constant Hamiltonian in spatially invariant variational problems, going back to Noether's theorem. However, we also give a specific interpretation of both the fourth-order differential equation and the Hamiltonian in terms of horizontal and vertical variations, with clear analogues with the concept of virtual work. Here horizontal and vertical variations define \textit{virtual displacements} on the system, and the resulting ODEs describe the required load balance at a given point of a stationary solution. In Section~\ref{sec:physicalinter} we show how integration of the Euler-Lagrange equation and the Hamiltonian gives vertical and horizontal force balances for the system, where individual terms can be identified with their physical counterpart.

Section~\ref{sec:existence} gives a shooting argument that shows there exists a unique solution to this obstacle problem. These can be rescaled to form a one-parameter group, which gives the main result of Section~\ref{sec:Scaling}:

\medskip
\begin{theorem}\label{th:Introduction}
Given $k>0$ so that $f = k|x|$, there exists a constant $\beta = \beta(k) > 0$ such that for all $q>0$ and $B>0$, the horizontal size of the void $\ell$ and the vertical shear force at the point of contact $Bw_{xxx}(\ell -)$ scales so that
\[
\ell = \beta \left(\frac{q}{B}\right)^{-1/3} \quad Bw_{xxx}(\ell -) = -B\frac{\beta}{(1 + k^2)^{5/2}}\left(\frac{q}{B}\right)^{2/3}
\]
\end{theorem}
\medskip

In Section~\ref{sec:Numerical} we show that these analytical results agree with the numerics, as well as with physical intuition. As the ratio of overburden pressure to bending stiffness becomes large, the size of the void tends to zero, giving a deformation with straighter limbs and sharper corners. By allowing the layers to form a void, the model is capable of producing both gently curving and sharp-cornered folds, without violating the elastic assumptions. Understanding this local behavior at areas of intensive folding may be seen as a first step to a multilayered model with the possibility of voiding.


\section{A voiding model close to a geometric singularity}
\label{sec:Model}

\subsection{The modelling}
We consider an infinitely long thin elastic layer, of stiffness~$B$, whose deformation is characterized by its vertical position $w(x)$ as a function of the horizontal independent variable $x\in\R$. Overburden pressure, from the weight of overlying layers, acts perpendicularly to the layer with constant magnitude~$q$ per unit length. The layer is constrained to lie above the a V-shaped obstacle, defined by the function $f(x) = k|x|$, i.e. $w\geq f$. 
Although we appear to solve the problem for an infinitely thin layer, the analysis is the same for any layer of uniform thickness up to changes of stiffnes $B$. In all cases $w(x)$ defines the ceterline of the layer, and $f(x)$ defines the shape the layer would take in the absense of voids. This is only possible in this special case since $f$ has straight limbs, and can therefore be propagated forwards and backwards without change. The setup and parameters of the model are summarised in Fig. \ref{fig:model}.

The \emph{contact set} of a function  $w$ is the set $\Gamma(w) = \{ x \in \Re : w(x) = f(x)\}$, the non-contact set  $\Gamma^c(w)$ is its complement, and we define the two contact limits $\ell_+ = \inf\{ x > 0: u(x) = f(x)\}$ and $\ell_- = \sup\{ x < 0: u(x) = f(x)\}$. 

\medskip
We now derive a total potential energy function for the system, described by the displacement $w$.

\subsubsection{Bending Energy}
Classic bending theory (e.g.~\cite[Ch.~1]{ThompsonHunt73}) gives the bending energy over a small segment of the beam $ds$ as $dU_B = \frac{B}{2}\kappa(s)^2ds$, where $\kappa$ is curvature. Integrating over all $s$ we find 
\[
U_B = \frac{B}{2}\int^{\infty}_{-\infty} \kappa^2 ds = \frac{B}{2}\int^{\infty}_{-\infty}\frac{w_{xx}^2}{(1 + w_x^2)^{3}} \frac{ds}{dx}dx = \frac{B}{2}\int^{\infty}_{-\infty} \frac{w_{xx}^2}{(1 + w_x^2)^{5/2}} dx
\]
The quadratic dependence on $w_{xx}$ implies that a sharp corner has infinite bending energy. This is the basic reason why at any finite overburden pressure the elastic layer will show some degree of voiding.

\subsubsection{Work done against overburden pressure in voiding}
The overburden pressure acting on the layer is $q$ per unit length, therefore considering displacements $w$ for which $w \geq f$ the work done by overburden pressure in voiding is given by $q(w - f)\,dx$, and integrating over all $x$ gives
\[
{
U_V = q\int^{\infty}_{-\infty}(w-f) dx.
}
\]
We see that if $q$ is large, then $U_V$ becomes a severe energy penalty.
\subsubsection{Total potential energy}
The total potential energy function is the sum of bending energy and work done against overburden pressure,
\begin{equation}
{
V = \frac{B}{2}\int^{\infty}_{-\infty} \frac{w_{xx}^2}{(1 + w_x^2)^{5/2}} dx 
+ q\int^{\infty}_{-\infty}(w-f)\,dx
}
\label{Energyfunctional}
\end{equation}
The solutions of the system are then minimizers of the energy functional (\ref{Energyfunctional}) subject to the constraint $w \geq f$. 

A natural space on which to define $V$ is the complicated-looking $H^2_{\mathrm{loc}}(\R)\cap (f+L^1(\R))$. Here $H^2_{\mathrm{loc}}(\R)$ is the space of all functions with second derivatives in $L^2(K)$ for any compact set $K\subset \R$. Finiteness of the first term in $V$ requires (at least) $w\in H^2_{\mathrm{loc}}(\R)$, and well-definedness of the second term requires that $w-f\in L^1(\R)$. However, under the conditions $w\geq f$ and $V(w)<\infty$ these conditions are automatically met, and therefore we will not insist on the space below.

\subsection{Constrained minimization of total potential energy}

\subsubsection{Properties and existence of minimizers}
Before deriving necessary conditions on minimizers of~\eqref{Energyfunctional} under the condition $w\geq f$, we first establish a few basic, but important, properties. These are that a constrained minimizer exists, is necessarily convex and symmetric, and has a single interval in which it is \emph{not} in contact with the obstacle. We will prove uniqueness using different methods in Section~\ref{sec:existence}.

\medskip
We write $w^{\#}$ for the convex hull of $w$, i.e. the largest convex function $v$ satisfying $v\leq w$. If $w\geq f$, then since $f$ is convex, it follows that $w^{\#} \geq f$. 

\medskip

\begin{theorem}\label{th:convexproof}
For any $w$, $V(w^\#)\leq V(w)$, and any constrained minimizer $w$ is convex. For all $x\in\R$, $-k\leq w_x(x)\leq k$.
\end{theorem}

\medskip
\begin{proof}
First we note that if $w\in H^2_{\mathrm{loc}}(\Re)$, also $w^\#\in H^2_{\mathrm{loc}}(\Re)$. Indeed, by considering expressions of the form 
\[
w^\#_x(x_2) - w^\#_x(x_1) = \int_{x_1}^{x_2} w^\#_{xx}(x)\, dx,
\]
it follows that the measure $w^\#_{xx}$ is Lebesgue-absolutely continuous, and satisfies $0\leq w^\#_{xx}\leq |w_{xx}|$. Since $w_{xx}\in L^2(\Re)$, it follows that $w^\#_{xx}\in L^2(\Re)$. Then $w^\#\in H^2(K)$ for all compact $K\subset \Re$ by integration. 

Defining the set $\Omega := \{x\in \Re: w^\#(x) = u(x)\}$, the function $w^\#$ is twice differentiable almost everywhere on $\Omega$, with a second derivative $w^\#_{xx}$ equal to $w_{xx}$ almost everywhere on $\Omega$. On the complement $\Omega^c$,  $w^\#_{xx}=0$ by~\cite[Theorem~2.1]{Griewank1990}.

Substituting $w^{\#}$ into (\ref{Energyfunctional}) shows that $V(w)\geq V(w^{\#})$, with equality only if $w^\#=w$. Since $w$ minimizes $V$, we have $w=w^\#$, and therefore $w$ is convex.

The restriction on the values of $w_x$ follows from the monotonicity of $w_x$ and the fact that $w-f$ tends to zero at $\pm\infty$.
\qquad\end{proof}

\medskip\noindent
As a direct consequence of Theorem \ref{th:convexproof},

\medskip
\begin{theorem}\label{th:intervalcontact}
The non-contact set $\Gamma^c(w)$ of a minimizer $w$ is an interval containing $x=0$, and for all $x \geq \ell_+$ and $x \leq \ell_-$ we have $w(x) = f(x)$.
\end{theorem}

\smallskip
Note that this statement still allows for the possibility that $\ell_\pm = \pm \infty$. 

\medskip
\begin{proof}
Suppose that $x_1,x_2> 0$ are such that $w=f$ at $x=x_1$ and at $x=x_2$. By convexity of $w$ we then have $w=f$ on the interval $[x_1,x_2]$. If the contact set $\Gamma(w)$ is bounded from above, then by the convexity of $w$, there exists $\varepsilon>0$ and $a>0$ such that $w(x)\geq a+ (k+\varepsilon)x$ for all $x\in \Re$, implying that $U_V(w)=\infty$. Therefore $\Gamma(w)\cap[0,\infty)$ is an interval, and if it is non-empty, then it is necessarily extends to $+\infty$. Similarly,  $\Gamma(w)\cap(-\infty,0]$ is an interval, and if non-empty it extends to $-\infty$.

Finally, note that $x=0$ can not be a contact point, since the condition $w\geq f$ would imply that $w\not\in H^2_{\mathrm{loc}}(\Re)$. Therefore the \emph{non}-contact set $\Gamma^c(w)$ is an interval that includes $x=0$. 
\qquad\end{proof}

\medskip
\begin{theorem}\label{th:symmetryproof}
Any minimizer $w$ is symmetric, so that $w(x) = w(-x)$.
\end{theorem}

\medskip
\begin{proof}
We proceed by using a cut-and-paste argument. If $w$ is a minimizer, then  it follows from Theorems \ref{th:convexproof} and \ref{th:intervalcontact} that $w$ is convex and for all $x \geq \ell_+$ and $x \leq \ell_-$, $w(x) = f(x)$. Therefore $w_x(\ell_{\pm}) = \pm k$, and the intermediate value theorem states that there exists $\hat{x} \in (\ell_-,\ell_+)$ such $w_x(\hat{x}) = 0$. If $V_{[-\infty,\hat{x}]}(w) \geq V_{[\hat{x},\infty]}(w)$, then  we define the function
\[
\tilde{w} = \left\{ \begin{array}{ll}
         w(x + \hat{x}) - k|\hat{x}| & \mbox{if $x < 0$};\\
         w(-(x + \hat{x})) - k|\hat{x}| & \mbox{if $x > 0$}.\end{array} \right.
\]

\noindent If $V_{[-\infty,\hat{x}]}(w) \leq V_{[\hat{x},\infty]}(w)$, then we define $\tilde{w}$ as

\[
\tilde{w} = \left\{ \begin{array}{ll}
         w(-(x + \hat{x})) - k|\hat{x}| & \mbox{if $x < 0$};\\
         w(x + \hat{x}) - k|\hat{x}| & \mbox{if $x > 0$}.\end{array} \right.
\]

\noindent In either case $\tilde{w} \in H^2_{\mathrm{loc}}$, $\tilde{w}$ is symmetric, and $V(\tilde{w}) \leq V(w)$. 

Since $\tilde w$ is a minimizer, $\tilde w$ solves a fourth-order differential equation in its non-contact set $\Gamma(\tilde w)^c$ (which includes $x=0$; see~\eqref{EL} and  Remark~\ref{rem:EL-on-R}). 
By standard uniqueness properties of ordinary differential equations (e.g.~\cite{Coddington1955}), $w$ and $\tilde w$ are identical on both sides of $x=0$, and remain such until they reach the constraint $f$. Therefore $w \equiv \tilde{w}$ and is therefore symmetric.
\qquad\end{proof}

\medskip
\begin{corollary}
Since $w$ is symmetric, $\ell_+ = -\ell_- = \ell$.
\end{corollary}

\medskip

Finally, these assembled properties allow us to prove the existence of minimizers:

\medskip

\begin{theorem}
There exists a minimizer of $V$ subject to the constraint $w\geq f$. 
\end{theorem}

\smallskip
\begin{proof}
Let $w_n$ be a minimizing sequence. By Theorems~\ref{th:convexproof} and~\ref{th:symmetryproof} we can assume that $w_n$ is convex and symmetric, and we therefore consider it defined on $\R^+$. By the convexity, since $w_n(x)-f(x)\to0$ as $x\to\infty$, the derivative $w_{n,x}$ converges to $f'(\infty) = k$ as $x\to\infty$; therfore the range of $w_{n,x}$ is $[0,k]$. Since by convexity $\int_0^\infty (w_n-f)\, dx \geq w_n(0)^2/{2k}$, the boundedness of $V(w_n)$ implies that $w_n(0)$ is bounded. 

From the upper bounds on $w_{n,x}$ it follows that $w_{n,xx}$ is bounded in $L^2(\R^+)$; combined with the bounds on $w_n(0)$ and $w_{n,x}(0)=0$, this implies that a subsequence converges weakly in $H^2(K)$ to some $w$ for all bounded sets $K\subset [0,\infty)$. Since therefore $w_{n,x}$ converges uniformly on bounded sets, it follows that $w_x(0)=0$ and that
\[
\liminf_{n\to\infty} \int^{\infty}_0 \frac{w_{n,xx}^2}{(1 + w_{n,x}^2)^{5/2}} dx
\geq \int _0^\infty\frac{w_{xx}^2}{(1 + w_x^2)^{5/2}} dx.
\]
Similarly, uniform convergence on bounded sets of $w_n$, together with positivity of $w_n-f$, gives by Fatou's Lemma
\[
\liminf_{n\to\infty} \int^{\infty}_0(w_n-f)\,dx\geq \int^{\infty}_0(w-f)\,dx.
\]
Therefore $V(w)\leq \liminf V(w_n)$, implying that $w$ is a minimizer. 
\qquad\end{proof}

\subsubsection{The Euler-Lagrange equation}

We now apply the Kuhn-Tucker theorem~\cite[pp. 249]{Luenberger1968} to derive necessary conditions for minimizers of (\ref{Energyfunctional}) subject to the constraint $w \geq f$. Since any minimizer $w$ is symmetric by Theorem~\ref{th:symmetryproof}, we restrict ourselves to symmetric $w$, and therefore consider $w$ defined on $\R^+$ with the symmetry boundary condition $w_x(0)=0$.  

\medskip
\begin{theorem}\label{th:KuhnTucker}
Let $q,B,k>0$.
Define the set of admissible functions
\begin{equation}
\label{def:adm}
\mathcal{A} = \left\{w \in f+ H^2(\Re^+)\cap L^1(\R^+):  w_x(0) = 0  \right\}.
\end{equation} 
If $w$ minimizes (\ref{Energyfunctional}) in $\mathcal A$ subject to the constraint $w\geq f$, then it satisfies the stationarity condition
\begin{equation}\label{eqn:weaksol}
\int_0^\infty\left[B\frac{w_{xx}}{(1 + w_x^2)^{5/2}}\varphi_{xx} - \frac{5}{2}B\frac{w_{xx}^2w_x}{(1 + w_x^2)^{7/2}}\varphi_x + q\varphi\right]\, dx = \int_0^\infty \varphi \, d\mu,
\end{equation}
for all $\varphi \in H^2(\R^+)\cap L^1(\R^+)$ satisfying $\varphi_x(0)=0$, where $\mu$ is a non-negative measure satisfying the complementarity condition $\int_0^\infty (w-f)\,d\mu = 0$. 
\end{theorem}

\medskip
\begin{proof}
For the application of the Kuhn-Tucker theorem we briefly switch variables, and move to the linear space $X:=H^2(\Re^+)\cap L^1(\R^+)$, taking as norm the sum of the respective norms of $H^2$ and $L^1$. For any $w\in \mathcal A$, we define the \textit{void function} $v := w-f$, which is an element of $X$; the two constraints $v_x(0) := w_x(0)- f_x(0+) = -k$ and $v := w-f \geq 0$ are represented by the constraint $\mathcal G(v)\leq 0$, where $\mathcal G: X \to Z := \Re\times \Re\times H^2(\Re^+)$ is given by
\[
\mathcal G(v) := \left(\begin{array}{c}
v_x(0) + k\\
- v_x(0) - k\\
-v
\end{array}\right).
\]
We also define $\hat V(v) := V(v+f)$.

\smallskip
If $w$ satisfies the conditions of the Theorem, then the corresponding function $v\in X$ minimizes $\hat V$ subject to $\mathcal G(v)\leq 0$. 
The functionals $\hat V: X \to \Re$ and $\mathcal{G}: X \to Z$ are Gateaux differentiable; since $\mathcal G$ is affine, $v$ is a \emph{regular point} (see~\cite[p.~248]{Luenberger1968}) of the inequality $\mathcal G(v)\leq 0$. The Kuhn-Tucker theorem~\cite[p.~249]{Luenberger1968} states that there exists a~$z^*$ in the dual cone $P^* =\left\{z^* \in Z^*:\left\langle z^*,z\right\rangle \geq 0\; \forall z\in Z \text{ with } z\geq 0 \right\}$ of the dual space~$Z^*$, such that the Lagrangian
\begin{equation}
{
\mathcal{L}(\cdot) := \hat V(\cdot) + \left\langle \mathcal{G}(\cdot),z^*\right\rangle
}
\label{KuhnTucker}
\end{equation}
is stationary at $v$; furthermore, $\left\langle \mathcal G(v),z^*\right\rangle = 0$. 

This stationarity property is equivalent to~\eqref{eqn:weaksol}. 
The derivative of $\hat V$ in a direction $\varphi\in X$ gives the left-hand side of~\eqref{eqn:weaksol}; the right-hand side follows from the Riesz representation theorem~\cite[Th.~2.14]{Rudin1987}. This theorem gives  two non-negative numbers $\lambda_1$ and $\lambda_2$ and a non-negative measure $\mu$ such $\left\langle (a,b,u),z^*\right\rangle = \lambda_1a + \lambda_2 b +  \int_{0}^\infty u\, d\mu$ for all $a,b\in \R$ and $u \in X$. Therefore $\langle\mathcal G'(\varphi),z^*\rangle = -\int \varphi\, d\mu$ for any $\varphi\in X$ with $\varphi_x(0)=0$.

In addition, the complementarity condition $\langle G(v),z^*\rangle=0$ implies $\int_0^\infty v \, d\mu = \int_0^\infty (w-f)\, d\mu = 0.$
\qquad\end{proof}

\medskip
This stationarity property allows us to prove the intuitive result that all minimizers make contact with the support $f$:

\medskip
\begin{corollary}
\label{cor:ellbounded}
Under the same conditions the non-contact set, $\Gamma(w)^c$, is bounded, i.e. $\ell<\infty$.
\end{corollary}

\medskip
\begin{proof}
Assume that the contact set $\Gamma(w)$ is empty, implying $\mu\equiv0$.
In~\eqref{eqn:weaksol} take $\varphi_n(x) := \varphi(x-n)$ for some $\varphi\in C_c^\infty(\R)$ with $\int\varphi\, dx\not=0$.
Since $w-f\in L^1$ and $w_{xx}\in L^2$, we have $w_x(x)\to k$ as $x\to\infty$; therefore, as $n\to\infty$, the translated function
\[
y \mapsto \frac{w_{xx}}{(1+w_x^2)^{5/2}}(y+n) 
\]
converges weakly to zero in $L^2$, implying that the first term in~\eqref{eqn:weaksol}, with $\varphi=\varphi_n$, 
\[
\int_0^\infty \frac{w_{xx}}{(1+w_x^2)^{5/2}}\varphi_{n,xx} \, dx
= \int_{-n}^\infty \frac{w_{xx}}{(1+w_x^2)^{5/2}}(y+n) \varphi_{xx}(y)\, dy
\]
vanishes in the limit $n\to\infty$. The second term vanishes for a similar reason. In the limit $n\to\infty$ we therefore find $q\int\varphi\,dx = 0$, a contradiction.
\qquad \end{proof}

\medskip
The boundedness of the non-contact set now allows us to apply a bootstrapping argument to improve the regularity of a minimizer $w$, and derive a corresponding free-boundary formulation:

\medskip
\begin{theorem} \label{th:fulltheorem}
Under the same conditions as Theorem~\ref{th:KuhnTucker}, the function $w$ has the regularity $w\in C^\infty(\Gamma(w)^c)\cap C^2(\R^+)$, $w_{xxx}$ is bounded, and $w_{xxxx}$ is a measure; the Lagrange multiplier $\mu$ is
given by
\begin{equation}\label{eq:mu}
\mu = q\ell\delta_{\ell} + q\HS(\,\cdot\, - \ell)\Lebesgue,
\end{equation}
where $\HS$ is the Heaviside function, and $\Lebesgue$ is one-dimensional Lebesgue measure. In addition, $w$ and $\mu$ satisfy
\begin{equation}\label{EL}
B\left[\frac{w_{xxxx}}{(1+w_x^2)^{5/2}} - 10\frac{w_xw_{xx}w_{xxx}}{(1+w_x^2)^{7/2}} - \frac{5}{2}\frac{w_{xx}^3}{(1+w_x^2)^{7/2}}
+ \frac{35}{2}\frac{w_{xx}^3w_x^2}{(1+w_x^2)^{9/2}}\right]+ q = \mu
\end{equation}
in $\R^+$. 

Finally, $w$ also satisfies the free-boundary problem consisting of equation~\eqref{EL} on $(0,\ell)$ (with $\mu=0$), with fixed boundary conditions
\begin{equation}
\label{bc:fixed}
w_x(0) = 0
\qquad \text{and} \qquad w_{xxx}(0) = 0,
\end{equation}
and a free-boundary condition at the free boundary $x=\ell$, 
\begin{equation}
\label{bc:free}
w(\ell) = k\ell, \qquad w_x(\ell) = k, \qquad\text{and} \qquad w_{xx}(\ell) = 0 .
\end{equation}
\end{theorem}

\medskip
Before proving this theorem we remark that by integrating (\ref{EL}) we can obtain slightly simpler expressions. From integrating~\eqref{EL} directly, and applying~\eqref{bc:fixed}, we find 
\begin{equation}
B\frac{w_{xxx}}{(1 + w_x^2)^{5/2}}-\frac{5}{2}B\frac{w_{xx}^2w_x}{(1 + w_x^2)^{7/2}}+  qx = qx\HS(x-\ell)\qquad\text{for all }x>0.
\label{firstIntEL}
\end{equation}
By substituting the free boundary conditions at $x = \ell$ into~\eqref{firstIntEL} we also find that the limiting values of $w_{xxx}$ at $x=\ell$ are given by
\begin{equation}\label{eqn:gamma}
w_{xxx}(\ell-) = - (1+k^2)^{5/2}\frac qB \ell,
\qquad w_{xxx}(\ell+) = 0.
\end{equation}
In addition, by multiplying (\ref{firstIntEL}) by $w_{xx}$ and integrating we also obtain
\begin{equation}
\label{secondIntEL-pre}
\frac{B}{2}\frac{w_{xx}^2}{(1 + w_x^2)^{5/2}}+ q(xw_x - w) = \frac{B}{2}w_{xx}(0)^2 - qw(0).
\end{equation}
Note that the right-hand side of (\ref{firstIntEL}) does not contribute to the the integral since $w_{xx} = 0$ for all $x \geq \ell$. Substituting the boundary conditions~\eqref{bc:free}, we derive the condition
\begin{equation}\label{eqn:secondintcondition}
\frac{1}{2}Bw_{xx}(0)^2 = qw(0), 
\end{equation}
so that the previous equation becomes
\begin{equation}
{
\frac{B}{2}\frac{w_{xx}^2}{(1 + w_x^2)^{5/2}}+ q(xw_x - w) = 0.
}
\label{secondIntEL}
\end{equation}

\medskip
{\em Proof of Theorem~\ref{th:fulltheorem}.} Once again we switch variables to the \textit{void function}, $v:=w-f$ and define the functions
\[
g = B\frac{v_{xx}}{(1 + (v_x + k)^2)^{5/2}} \qquad\text{and}\qquad
h = -\frac 52 B \frac{v_{xx}^2(v_x + k)}{(1 + (v_x + k)^2)^{7/2}},
\]
by~\eqref{eqn:weaksol} we make the estimate
\begin{align*}\label{eqn:shortweakform}
\int_{\R^+}g\varphi_{xx} = -\int_{\R^+}h\varphi_x + \int_{\Re^+}(\mu - q)\varphi 
&\leq \|h\|_2\| \varphi_x\|_2 + \|\mu-q\|_{TV}\| \varphi\|_\infty\\
&\leq C(\| \varphi_x\|_2 + \| \varphi_x\|_1),
\end{align*}
where the total variation norm $\|\nu\|_{TV}$ is defined by 
\[
\|\nu\|_{TV} := \sup \Bigl\{ \int_{\R^+} \zeta\, d\nu: \zeta\in C(\R^+), \ \|\zeta\|_\infty<\infty\Bigr\}.
\]

Setting $\varphi_x = \psi$, it follows that $g$ is weakly differentiable, and $g_x \in L^2 + L^\infty$. From Theorem \ref{th:intervalcontact} and Corollary~\ref{cor:ellbounded}, $v|_{(\ell,\infty)} \equiv 0 \Rightarrow g_{x}|_{(\ell,\infty)} = 0$ and therefore $g_x \in L^2$. By calculating $g_x$ explicitly, we may write
\begin{equation}
\label{eq:vxxx}
v_{xxx} = (1 + (v_x + k)^2)^{\frac{5}{2}}g_x + \underbrace{5\frac{v_{xx}^2(v_x + k)}{1 + (v_x + k)^2}}_{\in L^1}.
\end{equation}
Theorem~\ref{th:convexproof} shows that $(1 + (v_x + k)^2)^{5/2} \in L^{\infty}$, therefore $v_{xxx} \in L^1 $, so that $ v_{xx} \in L^{\infty}$, which in turn shows that $v_{xxx} \in L^2$ by~\eqref{eq:vxxx}. We also see that since
\begin{equation}
h_x = \underbrace{-\frac{2v_{xx}v_{xxx}(v_x + k)}{(1 + (v_x + k)^2)^{7/2}}}_{\in L^2} - \underbrace{\frac{v_{xx}^3}{(1 + (v_x + k)^2)^{7/2}} + 7\frac{v_{xx}^3(v_x + k)^2}{(1 + (v_x + k)^2)^{9/2}}}_{\in L^\infty},
\end{equation}
we have $h_x\in L^2$. 
We now look to similarly bound $v_{xxxx}$. 
In the sense of distributions, we have 
\begin{equation}\label{shortformEL}
g_{xx} = -h_x + \mu - q, 
\end{equation}
and since $h_x$ has bounded support, this is an element of $\mathcal M$, the set of measures with finite total variation. We can now write
\[
v_{xxxx} = \underbrace{(1 + (v_x + k)^2)^{5/2}}_{\text{continuous and bounded}}\underbrace{g_{xx}}_{\in\mathcal M} + \frac{5}{2}\frac{\overbrace{3v_{xxx}v_{xx}(v_x + k)}^{\in L^2}  + \overbrace{v_{xx}^3}^{\in L^\infty}}{(1 + (v_x + k)^2} - \overbrace{\frac{35}{2}\frac{v_{xx}^2(v_x + k)^2}{(1 + (v_x + k)^2)^{3/2}}}^{\in L^\infty}
\]
Since $v_{xxxx}$ has finite total variation, $v_{xxx}$ is bounded. Calculating ~\eqref{shortformEL} explicitly we find
\begin{multline}\nonumber
B\left[\frac{v_{xxxx}}{(1+(v_x + k)^2)^{5/2}} - 10\frac{(v_x +k )v_{xx}v_{xxx}}{(1+(v_x + k)^2)^{7/2}} - \frac{5}{2}\frac{v_{xx}^3}{(1+(v_x + k)^2)^{7/2}}\right] + \\  + B\left[\frac{35}{2}\frac{v_{xx}^3(v_x + k)^2}{(1+(v_x +k)^2)^{9/2}}\right] + q = \mu.
\end{multline}
Switching back to the orginal variable $w = v + f$ gives~\eqref{EL}. 

We now turn to~\eqref{eq:mu}. From the complementarity condition $\int (w-f)\, d\mu = 0$
we deduce that $\supp\mu\subset \Gamma(w)$. Theorem \ref{th:intervalcontact} shows that  $w = f$ on $[\ell,\infty)$, and substituting this directly into~\eqref{EL} shows that $\mu|_{(\ell,\infty)} = q\Lebesgue|_{(\ell,\infty)}$. This proves that $\mu$ has the structure 
\[
\mu = \alpha\delta_\ell + q\HS(\,\cdot \, - \ell) \Lebesgue
\]
some $\alpha\geq0$.
To determine the value of $\alpha$, take $\varphi\in C^\infty([0,\infty))$ with bounded support, and such that $\varphi\equiv 1$ in $[0,\ell+1]$. Then the weak Euler-Lagrange equation~\eqref{eqn:weaksol} reduces to $\alpha = q\ell$.

We now turn to the boundary conditions. The boundary condition $w_x(0)=0$ is encoded in the function space, and the natural boundary condition $w_{xxx}(0) = 0 $ follows by standard arguments. The conditions $w(\ell) = k\ell$, $w_x(\ell) = k$, and $w_{xx}(\ell)=0$ all follow from the contact at $x=\ell$.
\qquad$\square$

\begin{remark}
\label{rem:EL-on-R}
An identical argument gives that any minimizer on $\Re$, without assuming symmetry, satisfies the equation~\eqref{EL} on $\Re$.
\end{remark}

\section{Hamiltonian, intrinsic representation, and physical interpretation}\label{sec:physical}

In this section we pull together an number of key results. First we calculate the Hamiltonian for the system and discuss its interpretation in a static setting. We then show that both the Hamiltonian and the Euler-Largrange equation for the system can be represented in a combination of cartesian and intrinsic coordinates, which allows us to intepret both objects physically. This physical interpretation shows the correspondence between the rigorous mathematical description of the problem, seen in Section~\ref{sec:Model}, and a physical understanding of the system.

\subsection{The Hamiltonian}

There is a long history of applying dynamical-systems theory to variational problems on an interval. Elliptic problems can thus be interpreted as Hamiltonian systems in the spatial variable $x$~\cite{Mielke91}, implying that the Hamiltonian is constant in space. 

We apply the same view here. The conserved quantity $\H$, which we again call the Hamiltonian,  is obtained by considering stationary points of the Lagrangian $\L$ in~(\ref{KuhnTucker}) with respect to horizontal or `inner' variations $x \mapsto x +\epsilon\varphi$ for some $\varphi \in H^2$. This defines a perturbed function $w^\epsilon(x) := w(x+\e\varphi(x))$, whose derivatives are
\begin{align*}
w^{\epsilon}_x(x) &= w_x(x+\e\varphi(x))(1 + \epsilon\varphi_x(x)),\\
w^{\epsilon}_{xx}(x) &= w_{xx}(x+\e\varphi(x))(1 + \epsilon\varphi_x(x))^2 + w_x(x+\e\varphi(x))\epsilon\varphi_{xx}(x).
\end{align*}
The requirement that the Lagrangian $\mathcal L$ is stationary with respect to such variations gives the condition
\begin{multline}\label{eqn:hamiltonian}
B\left[\frac{w_{xxxx}w_x}{(1+w_x^2)^{5/2}} - 10\frac{w_{xxx}w_{xx}w_x^3}{(1+w_x^2)^{7/2}} - \frac{5}{2}\frac{w_{xx}^3w_x}{(1+w_x^2)^{7/2}}
+ \frac{35}{2}\frac{w_{xx}^3w_x^3}{(1+w_x^2)^{9/2}}\right]\\
 + (q -\mu) w_x = 0.
\end{multline}
Integrating this equation we find that the left-hand side of the expression
\begin{equation}
\label{eqn:hamiltonianIntegral}
B\frac{w_{xxx}w_{x}}{(1 + w_x^2)^{5/2}} - \frac{5}{2}B\frac{w_{x}^2w_{xx}^2}{(1 + w_x^2)^{7/2}} -\frac{B}{2}\frac{w_{xx}^2}{(1 + w_x^2)^{5/2}} + qw - k qx\HS(x - \ell)
 = 0,
\end{equation}
is constant in $x$, and the fact that it is zero follows from its value at $x=0$ and~\eqref{eqn:secondintcondition}. By analogy to the remarks above we call the left-hand side above the Hamiltonian. 

Note that equation~\eqref{eqn:hamiltonian} is equal to~\eqref{EL} times $w_x$. This is a well-known phenomenon in Lagrangian and Hamiltonian systems, and can be understood by remarking that the perturbed function $w^\e$ can be written to first order in $\e$ as $w+ \e \varphi w_x$; therefore this inner perturbation corresponds, to first order in $\e$, to an additive (`outer') perturbation of $\varphi w_x$. 

\subsection{Intrinsic representation}

Equations~\eqref{firstIntEL} and~\eqref{eqn:hamiltonianIntegral} can be written in terms of intrinsic coordinates,  characterized by the arc length $s$, measured from the point of symmetry $x=0$,  and the tangent angle $\psi$ with the horizontal. First we recall the relevant relations between the two descriptions:
\begin{gather}\label{eqn:curvaturedsdx}
\psi_s = \kappa = w_{xx}/(1 + w_x^2)^{3/2}, \quad ds/dx = (1 + w_x^2)^{1/2}, \\
\label{eqn:cossin}
\cos\psi = \frac{dx}{ds} =  1/(1 + w_x^2)^{1/2}, \quad \sin\psi = \frac{dw}{ds} = w_x/(1 + w_x^2)^{1/2}.
\end{gather}
First we rewrite the integrated Euler-Lagrange equation, (\ref{firstIntEL}), as
\begin{multline}\nonumber
B\frac{d}{dx}\left[\frac{w_{xx}}{(1 + w_x^2)^{3/2}}\right]+\left[\frac{1}{2}B\frac{w_{xx}^2}{(1 + w_x^2)^3}\frac{w_x}{\sqrt{1 + w_x^2}}+ qx\right](1 + w_x^2) = \\
=qx\HS(x-\ell)(1 + w_x^2),
\end{multline}
and apply~\eqref{eqn:curvaturedsdx} and~\eqref{eqn:cossin} to obtain
\begin{equation}\nonumber
{
B\frac{d}{dx}\left[\psi_s\right]+\left[\frac{1}{2}B\psi_s^2\sin\psi+ qx - qx\HS(x-\ell)\right]\sec\psi\frac{ds}{dx} = 0,
}
\end{equation}
which can also be written as
\begin{equation}
{
B\psi_{ss}\cos\psi+ \frac{1}{2}B\psi_s^2\sin\psi+ qx = qx\HS(x-\ell).
}
\label{firstIntELintrinsic}
\end{equation}
Similarly, the integral (\ref{secondIntEL}) may be represented as
\begin{equation}
{
\frac{1}{2}B\psi_s^2\cos\psi+ q(x\tan\psi - w) =0.
}
\label{secondIntELintrinsic}
\end{equation}
Following a similar process, the Hamiltonian (\ref{eqn:hamiltonianIntegral}) can be rearranged to
\[
\frac{w_{x}}{(1 + w_x^2)}\frac{d}{dx}\left[\frac{w_{xx}}{(1 + w_x^2)^{\frac{3}{2}}}\right] - \frac{1}{2}B\frac{w_{xx}^2}{(1 + w_x^2)^3}\frac{1}{(1 + w_x)^{\frac{1}{2}}} + qw = kqx\HS(x-\ell).
\]
In intrinsic coordinates this gives the expression
\begin{equation}
{
B\psi_{ss}\sin\psi - \frac{1}{2}B\psi_s^2\cos\psi + qw =  kqx\HS(x-\ell).
}
\label{hamiltonianintrinsic}
\end{equation}
Note that equations~\eqref{firstIntELintrinsic} and~\eqref{hamiltonianintrinsic} can be combined to give 
\begin{equation}
\label{eq:psiss}
B\psi_{ss} + qx\cos\psi + qw\sin\psi = 2qxH(x-\ell)\cos\psi.
\end{equation}

\subsection{Physical interpretation in terms of force balance}\label{sec:physicalinter}


The combination of Cartesian and intrinsic coordinates that we  have used allow us to identify terms of (\ref{firstIntELintrinsic}) and (\ref{hamiltonianintrinsic}) with their physical counterparts. 
Figure~\ref{fig:forcebalance} demonstrates the forces acting on a section of the beam, from $s=0$ to $s=s$, together with a conveniently chosen area of pressurized matter. Note that force balances are conveniently calculated for the solid object consisting of the beam \emph{and} the roughly triangular body of pressurized matter (indicated by the hatching); this setup allows us to identify the total horizontal and vertical pressure, exerted by $q$, as $qx$ and $q(w(x)-w(0))$.
\begin{figure}[htbp]
	\centering
			\includegraphics[width = 5in]{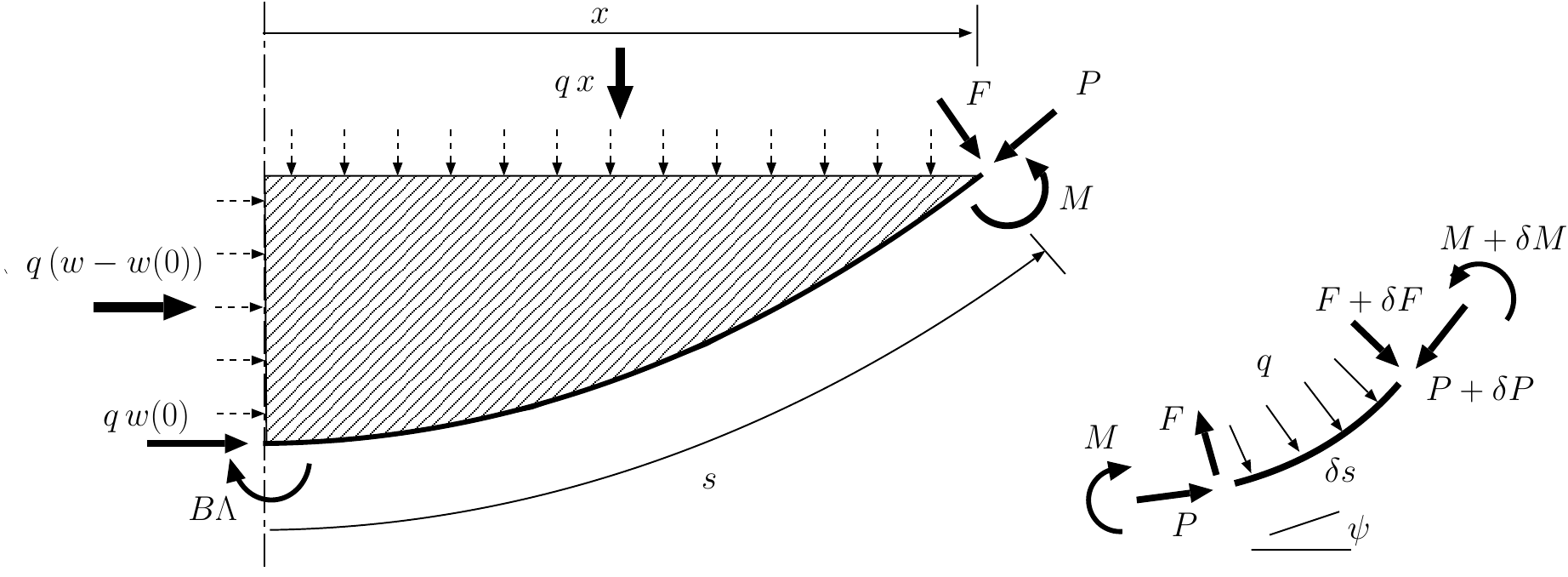}
	\caption{Left: forces on a section of the beam with pressurized matter. Right: small element.}
	\label{fig:forcebalance}
\end{figure}

The small element of the beam shown in Fig. \ref{fig:forcebalance} indicates how the well-known relations from small-deflection beam theory between lateral load $q$, shear force $F$,  and bending moment $M$, 
\[
{\rm d} F =q\, {\rm d} s \quad\mbox{ and }\quad {\rm d} M = F {\rm d} s,
\]
extend into large deflections. We now use these expressions to identify the terms of~\eqref{firstIntELintrinsic} and \eqref{hamiltonianintrinsic}.

The equilibrium equation~\eqref{EL} was obtained by additively perturbing $w$, i.e.\ by replacing $w$ by $w+\e \varphi$. This corresponds to a vertical displacement, which suggests that~\eqref{EL} can be interpreted as a balance of vertical load per unit of length. The integrated version~\eqref{firstIntEL} indeed describes a balance of the total vertical load on the rod from $s=0$ to $s=s$---i.e.\ the total vertical load on the setup in Fig.~\ref{fig:forcebalance}---as we now show. 

We write equation~\eqref{firstIntEL} in the intrinsic-variable version~\eqref{firstIntELintrinsic} as
\[
\underbrace{\overbrace{(B\psi_{s})_s}^{\text{shear force }F}\cos\psi}_{\text{vertical component of $F$}}
+ \underbrace{\overbrace{\frac{1}{2}B\psi_s^2}^{\text{axial load }P}\sin\psi}_{\text{vertical component of }P}{}+ \underbrace{qx}_{\substack{\text{total vertical}\\\text{ pressure}} }
={}\underbrace{qx\HS(x-\ell)}_{\substack{\text{total vertical}\\\text{ contact force}}}.
\]
Since by definition $M = B\psi_s$, the term $B\psi_{ss} = (B\psi_s)_s$ is the normal shear force $F$, and the first term above is its vertical component. The term $qx$ is the total vertical load exerted by the pressure $q$ between $s=0$ and $s=s$ (see Fig.~\ref{fig:forcebalance}), and $qxH(x-\ell)$ is the vertical component of the contact force. Finally, the only remaining force with a non-zero vertical component is the axial force $P$ at $x$, which can be interpreted as a Lagrange multiplier associated with the inextensibility of the beam. This suggests the interpretation of the only remaining term in the equation as the vertical component of $P$, implying that we can identify $P$ as
\begin{equation}
P = \frac12 B\psi_s^2.
\label{eq:P}
\end{equation}

We can do a similar analysis of the Hamiltonian equation~\eqref{eqn:hamiltonian}. Since this equation has been obtained by perturbation in the horizontal direction, we expect that integration in space gives an equation of balance of horizontal load. In the same way we write the integrated equation~\eqref{eqn:hamiltonianIntegral} in intrinsic coordinates (see~\eqref{hamiltonianintrinsic}) as
\[
\underbrace{(B\psi_{s})_s\sin\psi}_{\substack{\text{horizontal}\\\text{component of $F$}}} 
\quad  \underbrace{{}-\frac{1}{2}B\psi_s^2\cos\psi}_{\text{horizontal component of }P}{}+ \underbrace{q(w-w(0))}_{\substack{\text{total horizontal}\\\text{ pressure}} }
+ \underbrace{qw(0)}_{\substack{\text{horizontal load}\\\text{at $s=0$}} }
={}\underbrace{kqx\HS(x-\ell)}_{\substack{\text{total horizontal}\\\text{ contact force}}}.
\]
Then we similarly can identify the first two terms as the horizontal components of the shear force and the axial load, while the last term is the horizontal component of the contact force. The remaining two terms are the horizontal component of the pressure~$q$ and the axial force at $s=0$; the fact that this latter equals $qw(0)$ is consistent with~\eqref{eq:P} when one takes \eqref{eqn:secondintcondition} into account.

Note that the axial load $P$ of~\eqref{eq:P}, falling from a maximum compressive value at the centre of the beam to zero at the point of support, appears as a nonlinear term dependent on the bending stiffness $B$. Such terms are not normally expected in beam theory where, unlike for two-dimensional plates and shells, bending and axial energy terms are usually fully uncoupled. It comes about because of the re-orientation of the axial direction over large deflections. 

\section{Existence, uniqueness, and stability of solutions of the Euler-Lagrange equation}\label{sec:existence}

The Kuhn-Tucker theorem only provides a necessary condition for a minimizer; it provides no information about existence of one or many solutions, or about the stability of a solution. We now develop a shooting argument that proves both existence and uniqueness for the free-boundary problem~(\ref{EL}--\ref{bc:free}). This shooting method also motivates a numerical algorithm in Section~\ref{sec:Numerical}.

\medskip
\begin{theorem}\label{th:existence}
{
Given $q>0$, $B>0$, and $k>0$, there exists a function $w$ and a scalar $\ell>0$ that solve the free-boundary problem of Theorem~\ref{th:fulltheorem}.
}
\end{theorem}

\medskip
\begin{proof}
We consider the ODE (\ref{eq:psiss}) as an initial value problem with $\psi(0) = 0$ and $\psi_s(0) = \Lambda$, where $w$ is coupled to $\psi$ by~\eqref{eqn:cossin} and $w(0)=0$.  Since minimizers $w$ are convex (Theorem~\ref{th:convexproof}), we take $\Lambda>0$. For small $s>0$ we have $\psi\in (0,\pi/2)$, and therefore
\begin{align*}
\psi_s &= \Lambda - \int^s_0\biggl[\frac{q}{B}\sec(\psi(s'))x(s')+\frac{1}{2}(\psi_s(s'))^2\tan\psi(s')\biggr]ds' \\ 
&< \Lambda - \int^{x(s)}_0\frac{q}{B}\sec^2(\psi(s'))x'dx'<\Lambda -\frac{q}{2B}x^2.
\end{align*}
Hence, for all $\Lambda>0$ there is a point at $x=\ell(\Lambda) < \sqrt{2B\Lambda/q}$ at which $\psi_s = 0$ and therefore $w_{xx}(\ell) = 0$. From~(\ref{secondIntELintrinsic}) we deduce that
\[
\frac{1}{2}B\psi_s^2\cos\psi+ q(x\tan\psi - w) =\frac{1}{2}B\Lambda^2 - qw(0). 
\]
Therefore at $x=\ell$ we have
\[
\frac{1}{2}B\Lambda^2  + q(w - w(0)) = qxw_x, 
\]
and since $q(w - w(0)) > 0$ at $x = \ell$, and $x = \ell < \sqrt{2B\Lambda/q}$, it follows that
\[
\sqrt{\frac{B}{q}}\left( \frac{1}{2}\Lambda\right)^{3/2} < w_x(\ell).
\]
Now, consider the case of small $\Lambda$, so that $w$ is also small. To leading order we then have
\[
w_{xxx} + \frac{q}{B}x = 0, \quad w(0)=w_x(0)=0,\ w_{xx}(0)=\Lambda,
\]
so that
\[
w_{xx} = \Lambda - \frac{q}{2B}x^2, \quad\text{and} \quad  w_x = \Lambda x - \frac{q}{6B}x^3.
\]
This implies that if $\Lambda$ is sufficiently small,  then 
\[
w_x = \frac{2}{3} \Lambda \sqrt{\frac{2B}{q} \Lambda} < k,
\]
and conversely if $\Lambda$ is sufficiently large, then 
\[
w_x(\ell) > \sqrt{\frac{q}{B}}\left(\frac{1}{2}\Lambda\right)^{3/2}  > k.
\]
Hence, by continuous dependence of the solution $w$ on $\Lambda$, there is a value of $\Lambda$ and a value of $\ell$ for which
\[
w_x(\ell) = k \quad \mbox{and} \quad w_{xx}(\ell) = 0.
\]
If we now translate the function $w$ by adding the constant $k\ell-w(\ell)$, then the resulting function $w$ fulfills the assertion of the theorem.
\qquad\end{proof}

\medskip
We now show that this solution is in fact unique.

\medskip
\begin{theorem}\label{th:uniquenessproof}
The solution of the free-boundary problem of Theorem \ref{th:fulltheorem} is unique.
\end{theorem}

\medskip
\begin{proof}
The proof uses a monotonicity argument. Let $\psi(x,\Lambda)$ be a solution of~(\ref{eq:psiss})(written as a function of $x$) with $\psi_s(0)=\psi_x(0) = \Lambda>0$. Let $\Lambda_1 <\Lambda_2$; for small $x$, $\psi(x,\Lambda_1)<\psi(x,\Lambda_2)$. Let 
\[
\hat x := \sup\{x>0: \psi(x,\Lambda_1)<\psi(x,\Lambda_2)\} >0.
\]
Since  $w(x) - w(0) = \int^{x}_0\tan\psi$ it follows that 
\begin{equation}
w(x,\Lambda_1) - w(0,\Lambda_1) < w(x,\Lambda_2) - w(0,\Lambda_2),
\qquad\text{for all }0< x\leq \hat x.
\label{uniquenessproof1}
\end{equation}
Rewriting (\ref{secondIntEL-pre}) in terms of $\psi_x$ gives
\[
\psi_x = \psi_s\frac{ds}{dx} = \frac{w_{xx}}{(1 + w_x)^{5/2}} =  \sqrt{\frac{2}{B\cos^3\psi}\left[ \frac{1}{2}B\Lambda^2 + q(w-w(0)) - qx\tan\psi\right]}.
\]
Using~\eqref{uniquenessproof1} we deduce that for all $0< x\leq \hat x$, $\psi_x(x,\Lambda_1)<\psi_x(x,\Lambda_2)$, which implies that $\hat x = \infty$.

\medskip
Now assume that there exist two different solutions $\psi(x,\Lambda_1)$ and $\psi(x,\Lambda_2)$, with associated points of contact $\ell_1$ and $\ell_2$ such that $\psi(\ell_1,\Lambda_1) = \psi(\ell_2,\Lambda_2) = \arctan k$. Since we have shown that $\psi(\hat{x},\Lambda_2) > \psi(\hat{x},\Lambda_1)$, it follows that $\ell_2 < \ell_1$ (see Fig.~\ref{fig:unique}). Since $0< w_x(x,\Lambda_1)< k$ for all $0< x<\ell_1$, we have
\begin{equation}
\label{ineq:w1l1l2}
w(\ell_1,\Lambda_1) - w(\ell_2,\Lambda_1) < k (\ell_1-\ell_2).
\end{equation}
\begin{figure}[htbp]
	\centering\noindent
		\psfig{figure=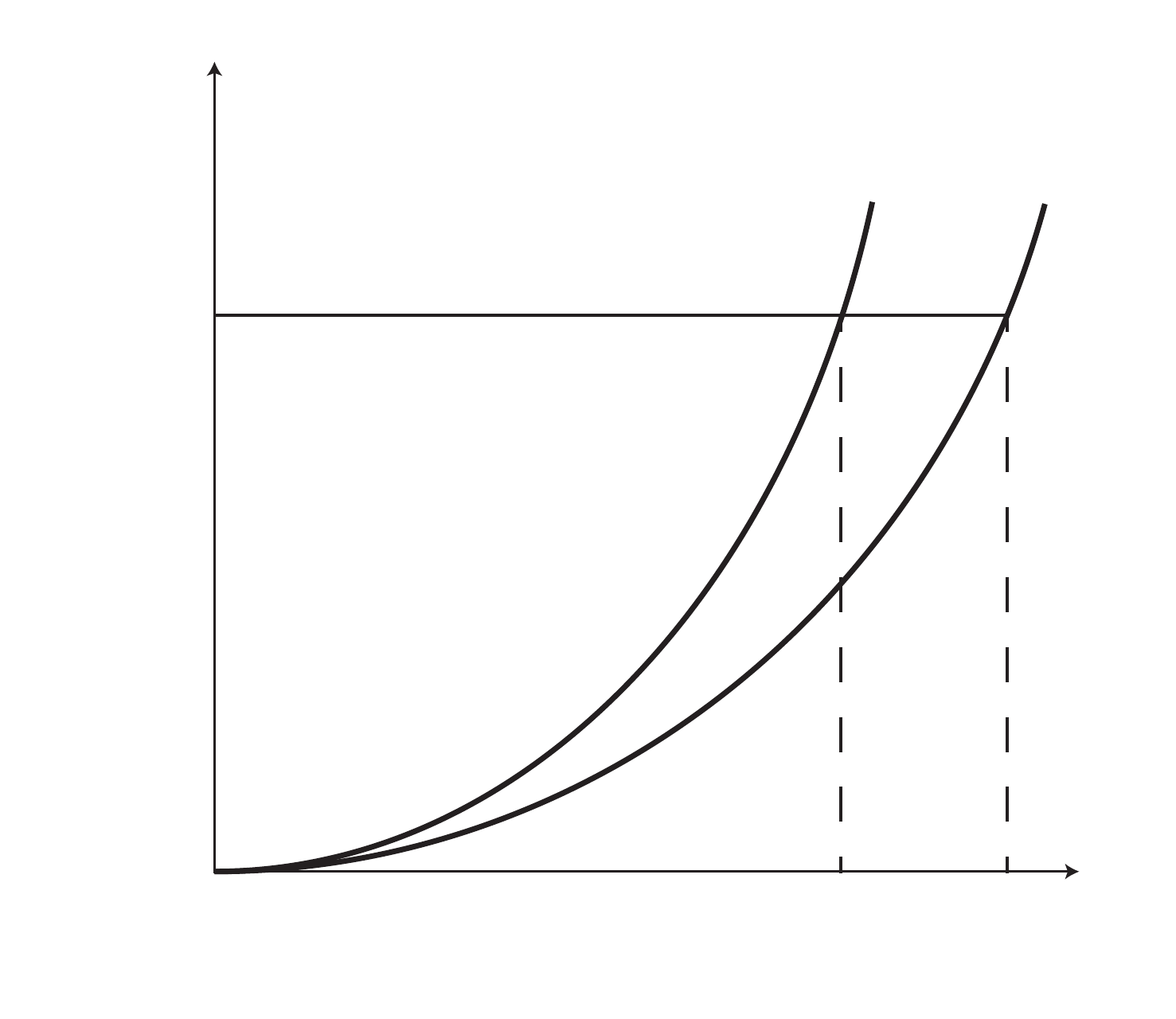,width=3in}
		\caption{The diagram shows the monotonicity argument used to prove uniqueness. If $\psi(\ell_1,\lambda_1) = \psi(\ell_2,\lambda_2) = \arctan k$ and $\psi(x,\Lambda_2) > \psi(x,\Lambda_1)$ for all $x>0$, then $\ell_2 < \ell_1$.}
	\label{fig:unique}
\end{figure}%
Evaluating~(\ref{secondIntELintrinsic}) at the free boundary for the solutions $\psi(\,\cdot\,,\Lambda_i)$ and the corresponding functions $w_i = w(\cdot,\Lambda_i)$ gives
\[
q(w_i(\ell_i)-w_i(0) + \frac{1}{2}B\Lambda_i^2 = qk\ell_i ,\qquad i=1,2.
\]
Writing the difference as
\begin{multline*}
q\bigl[(w_2(\ell_2)-w_2(0)) - (w_1(\ell_2)-w_1(0))\bigr]
+ \frac B2 (\Lambda_2^2-\Lambda_1^2)\\
+ q\bigl[k(\ell_1-\ell_2) - (w_1(\ell_1)-w_1(\ell_2))\bigr]
=0,
\end{multline*}
we observe from~\eqref{uniquenessproof1} and~\eqref{ineq:w1l1l2} that the left-hand side is strictly positive. This contradicts the assumption of multiple solutions.
\qquad\end{proof}



\section{Scaling Laws}\label{sec:Scaling}
We now see how the solutions of Section~\ref{sec:existence} can  be written as a one-parameter group parameterized by  $q/B$. Let $\ell(q,B,k)$ be the length of the non-contact set $\Gamma(w)$ of the solution $w$ for that $q$, $B$, and $k$, as defined in Section~\ref{sec:existence}.

\medskip

\begin{theorem}\label{th:ellscaling}
Given $k>0$, there exists a constant $\beta = \beta(k)>0$ such that for all $q>0$ and $B>0$, 
\begin{equation}
\label{eq:ell-scaling}
\ell(q,B,k) = \beta\left(\frac qB\right)^{-1/3}.
\end{equation}
\end{theorem}
\begin{proof}
If we let $x =: \lambda y$, $w =: \lambda u$, and $\ell =: \lambda\beta$, then the system (\ref{firstIntEL}) on $(0,\ell)$ rescales to
\begin{equation}\label{eqn:nondimfirstInt}
\frac{u_{yyy}}{(1 + u_y^2)^{5/2}}-\frac{5}{2}\frac{u_{yy}^2u_y}{(1 + u_y^2)^{7/2}}+ \lambda^3\frac{q}{B}y = 0 \qquad\text{on }(0,\beta).
\end{equation}
By choosing $\lambda$ such that $\lambda^3{q}/{B} = 1$, the problem reduces to that of finding a $w$ and~$\beta$ such that
\[
\frac{u_{yyy}}{(1 + u_y^2)^{5/2}}-\frac{5}{2}B\frac{u_{yy}^2u_y}{(1 + u_y^2)^{7/2}}+ y = 0,  \quad u_y(0)=0,u_y(\beta)=k, \mbox{ and } u_{yy}(\beta) = 0.
\]
Theorems \ref{th:existence} and \ref{th:uniquenessproof} prove that for each $k>0$ there exists a unique pair $(\beta, u)$ that solve (\ref{eqn:nondimfirstInt}). 
Transforming back, \eqref{eq:ell-scaling} follows.
\qquad \end{proof}

\medskip

Since $w_{xxx}(\ell) = -\frac qB (1+k^2)^{5/2}\ell.
$ (see~\eqref{eqn:gamma}), it follows that 

\smallskip

\begin{corollary}
The shear force $w_{xxx}(\ell-)$ satisfies
\[
w_{xxx}[q,B,k](\ell-) = -\frac{\beta}{ (1+k^2)^{5/2}} \left(\frac{q}{B}\right)^{2/3}.
\]
\end{corollary}

\section{Numerical results}\label{sec:Numerical}
Here we provide some numerical results to support the analytical results seen in the previous section. The shooting method of the previous section suggests a numerical algorithm, by reducing the boundary value problem to an initial value problem, and shooting from the free boundary with the unknown parameter $\ell$. A one parameter search routine was made using \textsc{matlab}'s built-in function \texttt{fminsearch}, which is an unconstrained nonlinear optimization package that relies on a modified version of the Nelder-Mead simplex method~\cite{Lagarias1998}. 

Finding global minimizers in an unknown energy landscape can lead to an unstable routine; however in this case the linearized version of~\eqref{EL} admits an analytic solution which provides a sufficiently accurate initial guess. Over the non-contact region equation~\eqref{EL} linearizes to
\begin{equation}
{
w_{xxxx} + \frac{q}{B} =0.
}
\label{eq:linearEL}
\end{equation}
Integrating (\ref{eq:linearEL}) and applying the boundary conditions at the free boundary $x =\ell$ gives the solution
\begin{equation}\nonumber
{
w = -\frac{1}{24}\frac{q}{B}x^4 + \frac{1}{2}\Lambda x^2 + w(0),
}
\label{eq:linearsolution}
\end{equation}
with the closed-form solution for $\ell$,
\begin{equation} \nonumber
\ell = \left(\frac{1}{3k}\frac{q}{B}\right)^{-\frac{1}{3}}
\end{equation}
Figure~\ref{fig:solutionprofiles} shows examples of solution profiles obtained in this manner. 

\begin{figure}[htbp]
\begin{center}
\noindent
\psfig{figure=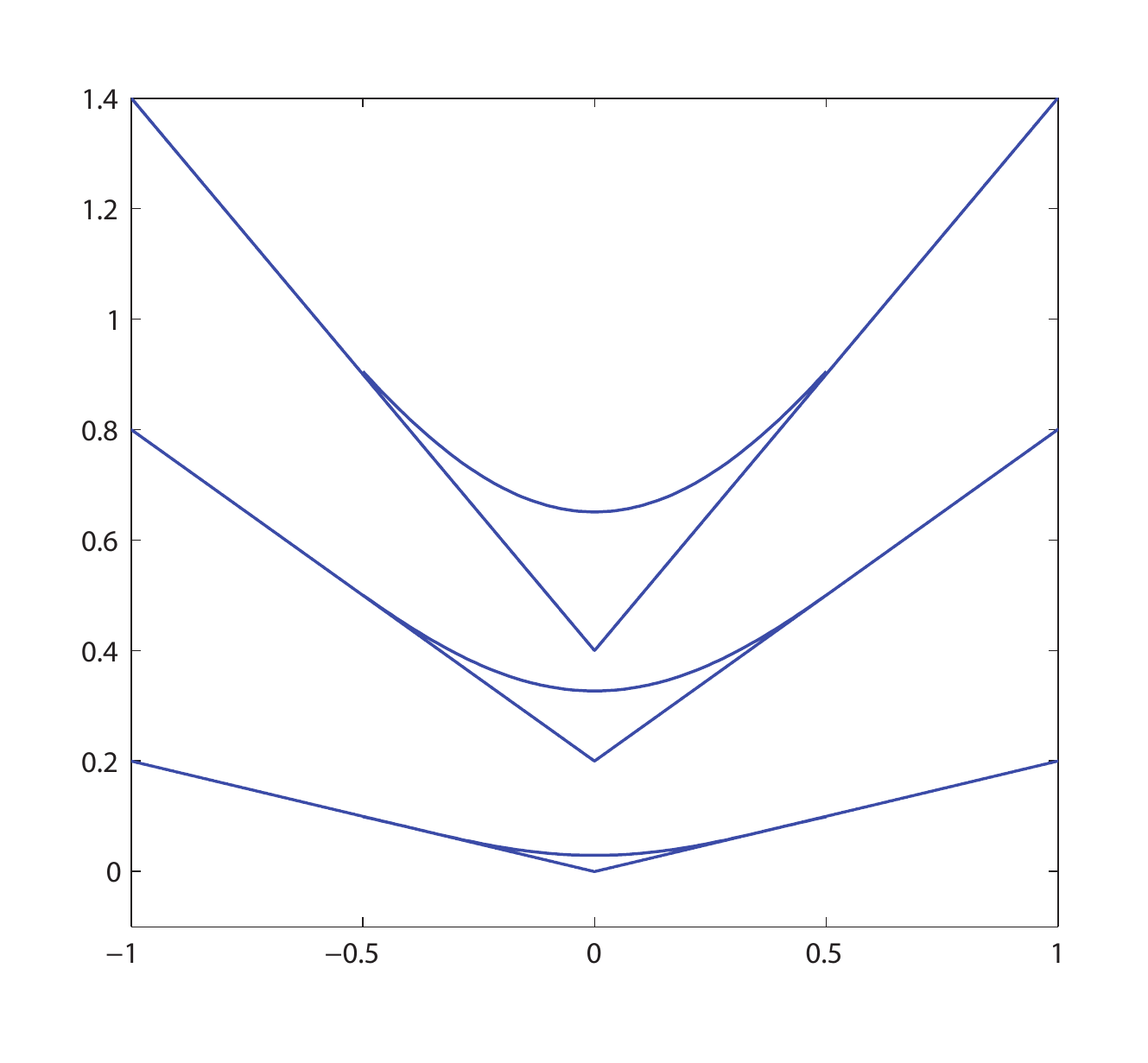,width=3.5in}
\caption{Solution profiles for fixed $q=1$, $B = 1$ and for increasing values of $k$}
\label{fig:solutionprofiles}
\end{center}
\end{figure}

\afterpage{\clearpage}
For fixed $k$, the parameters  $\ell$ and $w_{xxx}(\ell -) = -q\ell$ can be calculated numerically for varying values of $q/B$, and the results are shown in Fig.~\ref{fig:scaling}. These numerical results agree with the behaviour expected. For fixed $B$, increasing $q$ decreases the size of the delamination, yet increases the vertical component of shear at delamination, $Bw_{xxx}(\ell -)$. Numerically fitting these curves to a function of the form $\beta\left(\frac{q}{B}\right)^{\alpha}$, we see that the results agree with the scaling laws found in the previous section, so that
\[
\ell = \beta\left(\frac qB\right)^{-1/3},\quad
w_{xxx}(\ell -) = -(1+k^2)^{5/2}\beta\left(\frac{q}{B}\right)^{2/3}.
\]
\begin{figure}[htbp]
\centering
\noindent
\psfig{figure=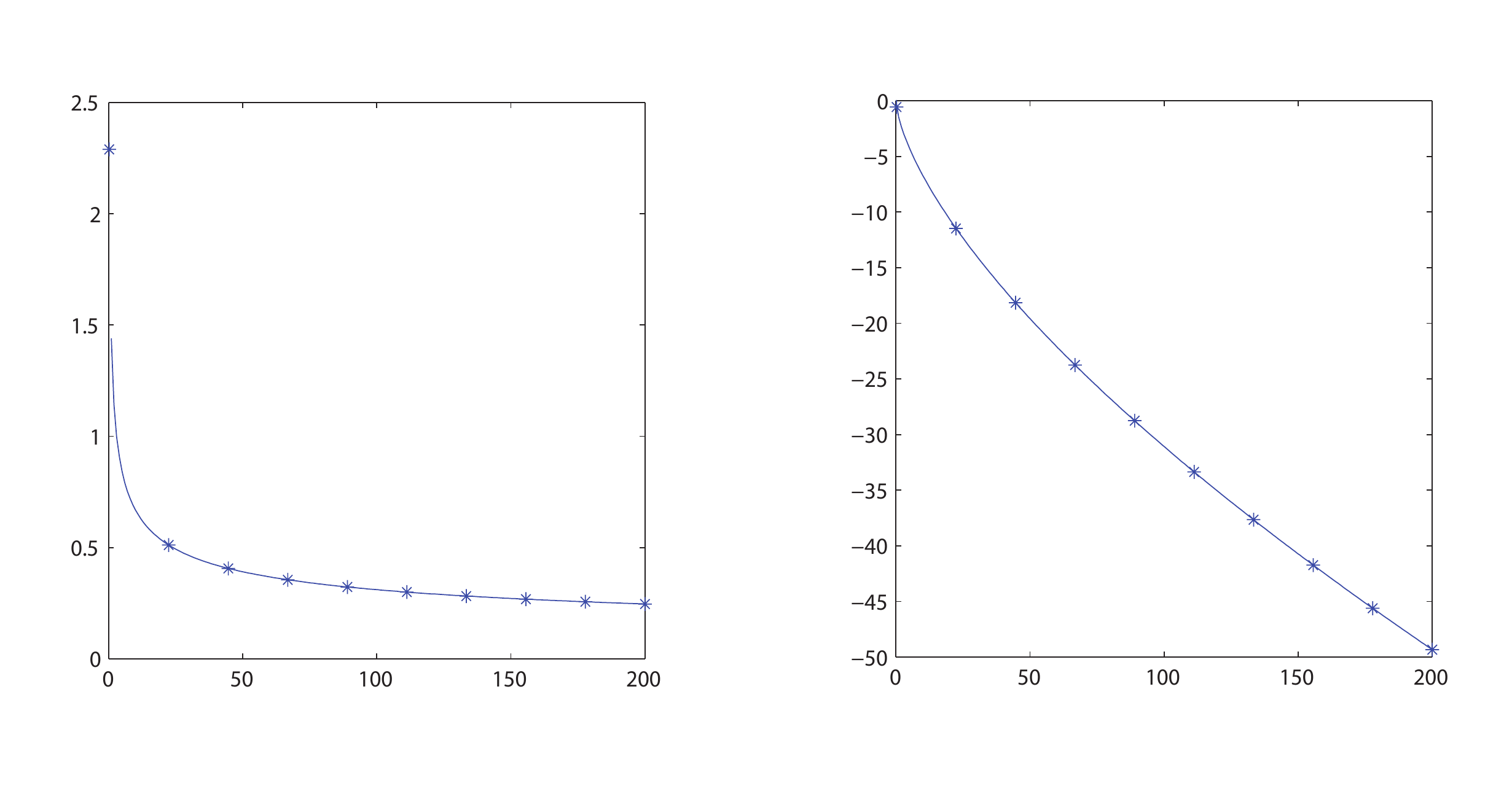,width=5in}	
\caption{Numerical results supporting the scaling laws found for $\ell$ and $w_{xxx}(\ell-)$ in Section~\ref{sec:Scaling}, results are shown for a fixed value of $k = 0.75$. (a) *'s show results found numerically for $\ell$ against $\frac qB$, the line plots $\beta\left(\frac qB\right)^{-1/3}$ (b) *'s show results found numerically for $w_{xxx}(\ell -)$ against $\frac qB$ , the line plots $-\frac{\beta}{(1 + k^2)^{5/2}}\left(\frac qB\right)^{2/3}$.}
\label{fig:scaling}
\end{figure}
Finally, Fig.~\ref{fig:kvsbeta} illustrates the dependence of $\beta$ on $k$.

\begin{figure}[htbp]
\centering
\noindent
\psfig{figure=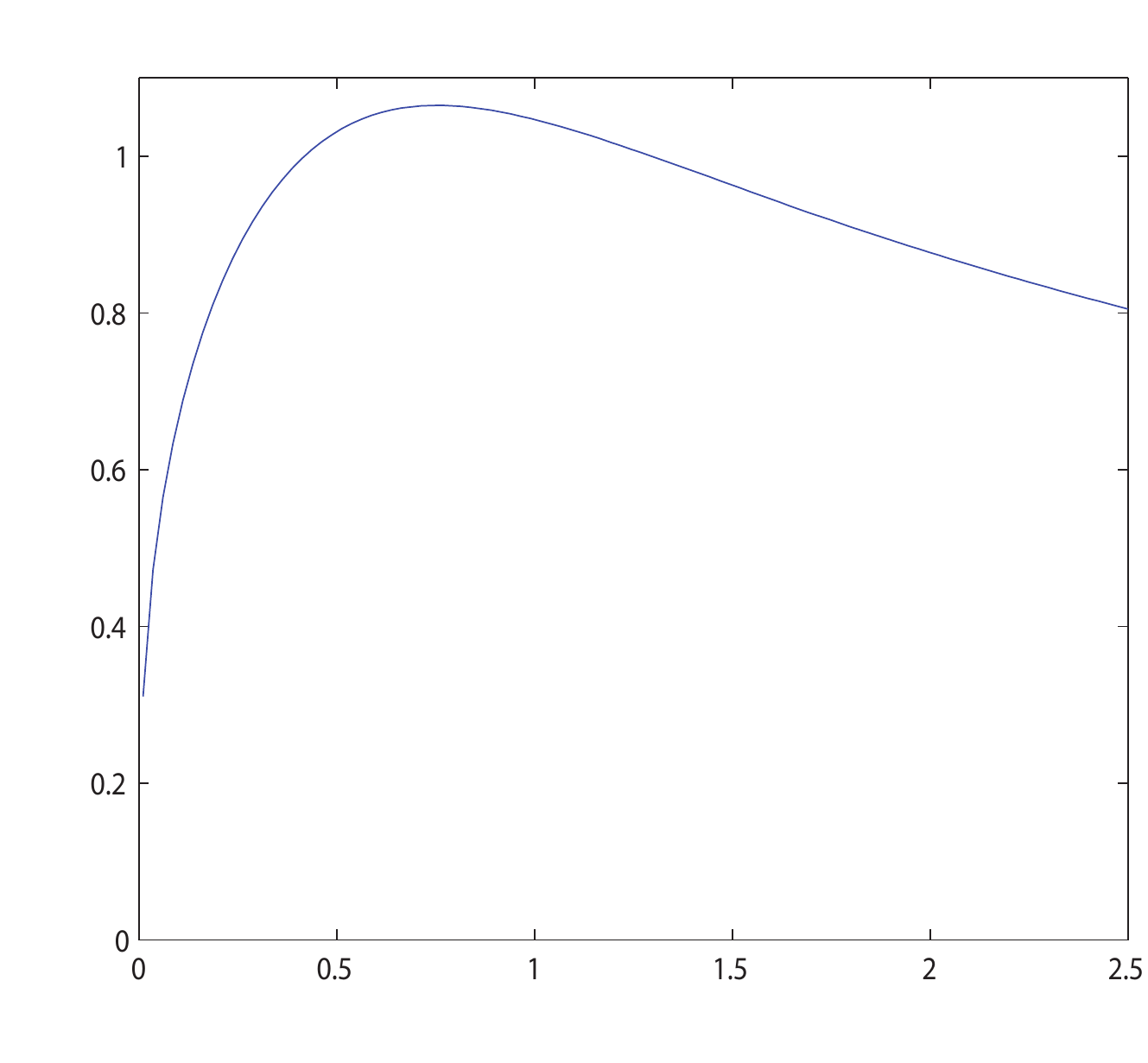,width=3in}
\caption{$\ell = \beta(k)$ versus $k$. Here $q=1$ and $B=1$.}
\label{fig:kvsbeta}
\end{figure}


\section{Concluding remarks}\label{sec:ConcludingRemarks}

The results of this paper show how elasticity, overburden pressure, and a V-shaped obstacle work together to produce one of the hallmarks of geological folds: straight limbs connected by smooth hinges. The model also gives insight into the relationship between material and loading parameters on one hand and the geometry and length scales of the resulting folds on the other. 

The model is of course highly simplified, and many modifications and generalizations can be envasiged. An important assumption is the pure elasticity of the material, and there are good reasons to consider other material properties of the layers. However, we believe the more interesting questions lie in other directions.

One such question is role of friction \emph{between} the layers, which was shown to be essential in other models of multilayer folding~\cite{HuntWadeePeletier99,HuntPeletierWadee00,BuddEdmundsHunt03,Wadee2004,Edmunds2006}. Since the normal stress between the layers changes over the course of an evolution, the introduction of friction will necessarily introduce history dependence, and the current energy-based approach will not apply. In this case other factors will also influence the behaviour, such as the length of the limbs, which determines the total force necessary for interlayer slip. 

An even more interesting direction consists in replacing the rigid, fixed, obstacle by a stack of other layers, i.e.\ by combining the multi-layer setup of~\cite{Boon2007} with the elasticity properties of this paper. 
A first experiment in that direction could be a stack of identically deformed elastic layers. An elementary geometric argument suggests that reduction of total void space could force such a stack in to a similar straight-limb, sharp-corner configuration, as illustrated in Fig.~\ref{fig:Chevron}.
\begin{figure}[ht]
\centering
\includegraphics[height=2.5cm]{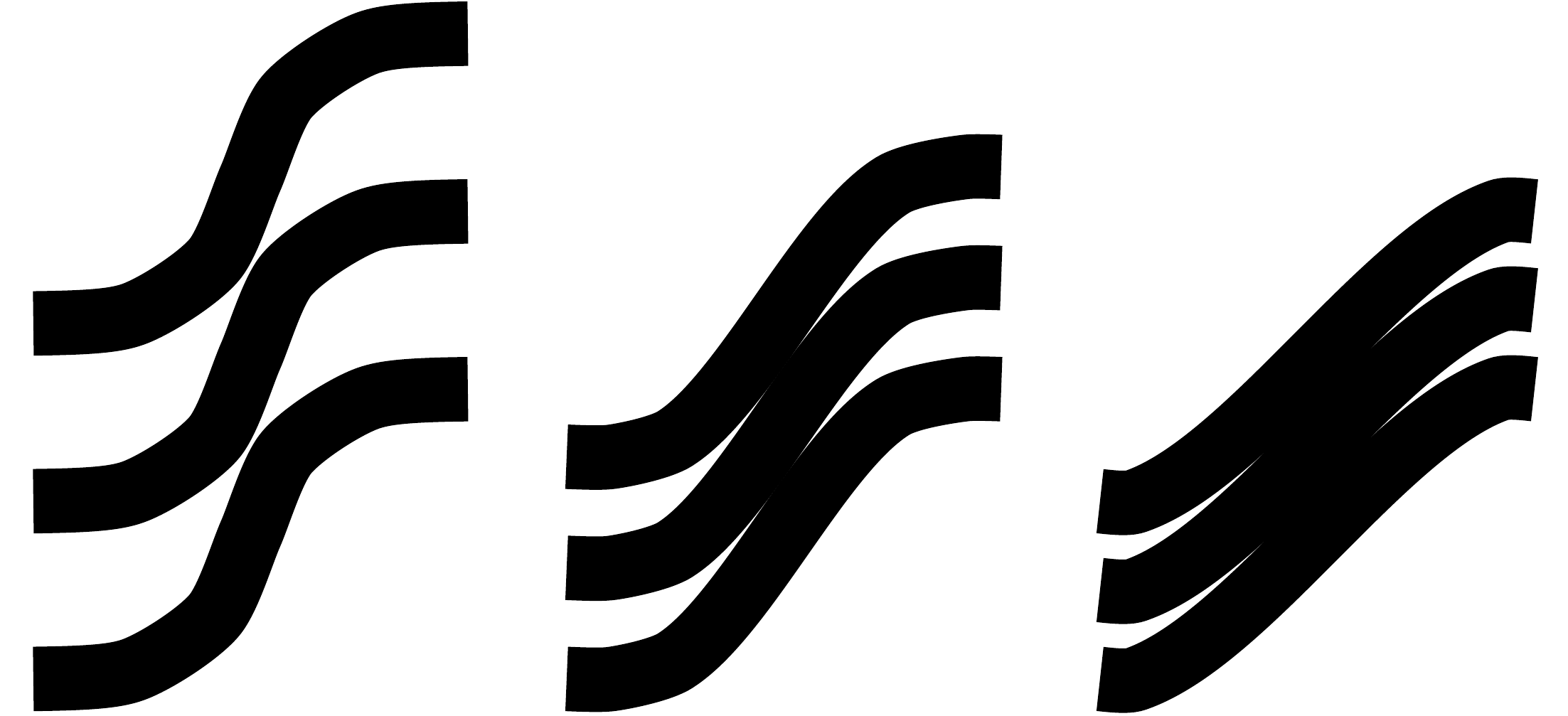}
\caption{Sharp-angle, straight-limb folds give rise to fewer voids than rounded folds (after~\cite{HuntPeletierWadee00}).}
\label{fig:Chevron}
\end{figure}
This suggests that this phenomenon should also be visible in a stack of compressed layers, and we plan to consider this problem in future work~\cite{Budd2011,Dodwell2011}.

\bibliographystyle{siam}
\bibliography{ref}

\end{document}